\documentclass[draft]{article}
\usepackage{amsthm}
\usepackage{amssymb}
\usepackage{amsmath}
\usepackage{array}
\usepackage[none,bottom]{draftcopy}

\usepackage[final]{graphicx}
\usepackage{hyperref}

\newtheorem{dfn}{Definition}
\newtheorem{theo}{Theorem}
\newtheorem{corollary}{Corollary}
\newtheorem{lemma}{Lemma}

\DeclareMathOperator{\dH}{d_H}
\DeclareMathOperator{\wH}{w_H}

\newcommand{\ignore}[1]{}

\newcommand{\GF}{\ensuremath{\mathbb{F}_2}}
\newcommand{\sen}[2]{s_#1(#2)}
\newcommand{\avsen}[1]{s_{#1}}
\newcommand{\Expec}[1] {\ensuremath{{\mathbb E}(#1)}} 
\newcommand{\vect}[1]{\ensuremath{{\mathbf{#1}}}}
\newcommand{\vectcomp}[2]{\ensuremath{{\mathrm {#1}_{#2}}}}
\newcommand{\BN}{ {\bf B}~}
\providecommand{\EX}{\mathbb E}

\providecommand{\set}[1]{\left\{ #1 \right\}}
\providecommand{\setcard}[1]{\left\lvert #1 \right\rvert}
\providecommand{\field}[1]{\mathbb #1}
\providecommand{\indexedvect}[2]{\vect{#1}^{(#2)}}
\providecommand{\indexedvectcomp}[3]{\vect{#1}^{(#2)}_{#3}}
\providecommand{\unitvect}[1]{\indexedvect{u}{#1}}

\newenvironment{keywords}%
   {\begin{trivlist}\item[]{\bfseries\sffamily Keywords:}\ }
   {\end{trivlist}}

\begin{document}

\title{Analysis of random Boolean networks using the average sensitivity}

\author{Steffen Schober\thanks{Corresponding author. E-Mail: Steffen.Schober@uni-ulm.de}\; and Martin Bossert\\
\small Institute of Telecommunications and Applied Information Theory, Ulm University\\
\small Albert-Einstein-Allee 43, 89081 Ulm, Germany
}

\maketitle

\begin{abstract}
In this work we consider random Boolean networks that provide a general model for {\em genetic regulatory networks}.
We extend the analysis of 
James Lynch who was able to proof Kauffman's conjecture that in the {\em ordered phase} of random networks, the number of 
{\em ineffective} and {\em freezing} gates is large, where as in the {\em disordered phase} their number is small.
Lynch proved the conjecture only for networks with connectivity two and non-uniform probabilities 
for the Boolean functions. 
We show how to apply the proof to networks with arbitrary connectivity $K$ and to 
random networks with {\em biased} Boolean functions.
It turns out that in these cases Lynch's parameter $\lambda$ is equivalent to the expectation of {\em average sensitivity} of the Boolean functions
used to construct the network. Hence we can apply a known theorem for the expectation of the average sensitivity.
In order to prove the results for networks with biased functions, we deduct the expectation of the
average sensitivity when only functions with specific connectivity and specific bias are chosen at random.
\end{abstract}


 
\section{Introduction}
In 1969 Stuart Kauffman started to study random Boolean networks as simple models of
genetic regulatory networks \cite{Kauffman1969}. Random Boolean networks  that consists
of a set of Boolean {\em gates} that are capable
of storing a single Boolean value. At discrete time steps these gates store a new value
according to an initially chosen random Boolean function, which receives its inputs
from random chosen gates. We will give a more formal definition later. Kauffman
made numerical studies of random networks, where the functions
are chosen from the set of all Boolean functions with $K$ arguments (the so called
{\em $NK$-Networks}).  He recognised that if $K\leq 2$, 
the random networks exhibit a remarkable form of ordered
behaviour: The limit cycles are small, the number of {\em ineffective gates},
which are gates that can be perturbed without
changing the asymptotic behaviour, and the number of {\em freezing gates} that stop changing
their state is large. In contrast if $K \geq 3$, the networks do not exhibit this kind of ordered behaviour
(see \cite{Kauffman1969,Kauffman1974}). The first analytical proof for this {\em phase transition} was 
given by Derrida and Pomeau (see \cite{Derrida1986}) by studying the evolution 
of the Hamming distance of random chosen initial states by means of so called {\em annealed approximation}.
The first proof for the number of freezing and ineffective gates was given by James Lynch (see
\cite{Lynch2005}, although slightly weaker results  appeared earlier \cite{Lynch1995,Lynch2002}).
Depending on a parameter $\lambda$, that depends 
on the probabilities of the Boolean functions, he showed that if $\lambda \leq 1$ 
almost all gates are ineffective and freezing, otherwise not.
Although his analysis is very general, until now it was only applied to networks
with connectivity $2$ and non-uniform probabilities for the Boolean function:
if the probability of choosing a constant function
is larger or equal the probability of choosing a non-constant non-canalizing function (namely
the XOR- or the inverted XOR-function),
$\lambda$ is less or equal to one.
But it turns out that in some cases $\lambda$ is equal to the expectation of the average sensitivity.
Therefore we will first study the average sensitivity in Section \ref{sec:02}.
Afterwards it will be shown in Section \ref{sec:03} how to use the results from the previous section to 
apply Lynch's analysis to classical $NK$-Networks and {\em biased} random Boolean networks \footnote{a definition will be given later}.
But first we will give some basic definition used throughout the paper in Section \ref{sec:01}.

\section{Basic Definitions}
\label{sec:01}
In the following $\GF = \set{0,1}$ denotes the Galois field of two elements,
where addition, denoted by $\oplus$, is defined modulo 2.
The set of vectors of length $K$ over $\GF$ will be denoted by $\GF^K$. If $\vect{x}$ is a vector from 
$\GF^K$, its $i$th component will be denoted by $\vectcomp{x}{i}$.
With $\unitvect{i} \in \GF^K$ we will denote the {\em unit vector} which has all 
components zero except component $i$ which is one.
The Hamming weight of $\vect{x} \in \GF^K$ is defined as 
\begin{equation*}
\wH(\vect{x}) = \setcard{\set{ i  \,|\, \vectcomp{x}{i} \neq 0,\, i=1, \dots,K}}
\end{equation*}
and the Hamming distance of 
 $\vect{x},\vect{y} \in \GF^K$ as 
\begin{equation*}
\begin{split}
\dH(\vect{x},\vect{y}) &= \wH( \vect{x} \oplus \vect{y}).
\end{split}
\end{equation*}
A Boolean function is a mapping $f:\GF^K \rightarrow \GF$. A function $f$ may be represented by its
{\em truth table} $\vect{t}_f$, that is, a vector in $\GF^{2^K}$, where each component of the truth table gives the value 
of $f$ for one of the $2^K$ possible arguments. 
To fix an order on the components of the truth table, suppose that its $i$th component 
equals the value of the corresponding function, given the binary representation (to $K$ bits) of $i$ as an argument.

\section{Average Sensitivity}
\label{sec:02}
 In this section we will focus on the {\em average sensitivity}.
The average sensitivity is a known complexity measure for 
Boolean functions, see for example \cite{Wegener1987}\footnote{here it is called {\em critical complexity}}. 
It was already used to study Boolean and random Boolean networks for example  in
\cite{Shmulevich2004,Kauffman2004}.

\begin{dfn} Let $f$ denote a Boolean function $\GF^K \rightarrow \GF$ and $\unitvect{i}$ a unit vector. 
\begin{enumerate}
\item
The sensitivity $\sen{f}{\vect{w}}$ is defined as: 
\begin{align*}
\sen{f}{\vect{w}} &= \setcard{\left\{ i \,|\, f(\vect{w}) \neq f(\vect{w} \oplus \unitvect{i}), i = 1,\dots, K\right\}}.
\end{align*}
\item 
The average sensitivity $\avsen{f}$ is defined as the average of $\sen{f}{\vect{w}}$ over all $\vect{w} \in \GF^K$:
\begin{equation*}
\avsen{f} = 2^{-K} \sum_{\vect{w} \in \GF^K} \sen{f}{\vect{w}}
\end{equation*}
\end{enumerate}
\end{dfn}

Now consider the random variable $F_K:\Omega \rightarrow {\cal F_K}$, 
where ${\cal F_K}$ denotes the set a all $2^{2^K}$ Boolean function with $K$ arguments.
The probability measure is given by $P(F_K=f)=\frac{1}{2^{2^K}}$.
The expected value of the average sensitivity of this random variable is denoted by
$ \EX_{F_K}(\avsen{f}) $, and is given by
\begin{equation*}
\EX_{F_K}(\avsen{f}) = \sum_{f} P(F_K=f) \avsen{f}
\end{equation*}
The expected value was already derived in \cite{Bernasconi1998}, and is given by:
\begin{theo}[Bernasconi \cite{Bernasconi1998}]
\hfill\\
Let the random variable $F_K$ be defined as above, then 
\begin{equation*}
\EX_{F_K}(\avsen{f}) = \sum_{f} P(F_K=f) \avsen{f} = \frac{K}{2}.
\end{equation*}
\label{theo:01}
\end{theo}

We will now concentrate on biased Boolean functions.
The bias of a Boolean function $f:\GF^K \rightarrow \GF$
is defined as the number of $1$ in the functions truth table divided by ${2^K}$.
To define the bias of a random Boolean function two definitions are possible.
First we can
assumes that the truth tables of the Boolean functions are produced by independent Bernoulli trials
with probability $p$ for a one (This should be called {\em mean} bias, used for example in \cite{Derrida1986,Shmulevich2004} ). 
Therefore consider the random variable $F_{K,p}$. The probability of choosing 
a function $f$ is given by
\begin{equation*}
P( F_{K,p} = f ) = p^{\wH(\vect{t}_f)} (1-p)^{2^K - \wH(\vect{t}_f)}
\end{equation*} 
For $p=1/2$ this is equivalent to the definition of $F_K$. 

As a second possibility,  we can only choose functions which have bias $p$ whereas to all other functions 
we assign probability 0 (we will call this {\em fixed} bias). 
Therefore consider the random variables $F^\text{fixed}_{K,p}: \Omega \rightarrow {\cal F_K}$.
Denote the truth table of a function $f$ by $\vect{t}_f$.
Further denote the set of all Boolean functions $f$ with $K$ arguments and $\wH(\vect{t}_f) = p 2^K$ 
with ${\cal F}_{K,p}$.
The probability for a certain function chosen according  $F^\text{fixed}_{K,p}$ is given by
\begin{equation*}
P( F^\text{fixed}_{K,p} = f ) = \begin{cases}
		\frac{1}{\setcard{{\cal F}_{K,p}}} & \text{if $f \in {\cal F}_{K,p}$} \\
		0 & \text{if $f \notin {\cal F}_{K,p}$}
	      \end{cases}
\end{equation*} 

Both definitions ensure 
that  the expectation to get a one is equal to $p$ if the input of a function is chosen 
at random (with respect to uniform distribution).
But it will turn out that these two different methods of creating biased Boolean functions, 
have a major impact on the average sensitivity.

The expectation of the average sensitivity of $F_{K,p}$ was derived in
\cite{Shmulevich2004}:
\begin{theo}[\cite{Shmulevich2004}]
Let the random variable $F_{K,p }$ be defined as above: 
\begin{equation*}
\EX_{F_{K,p}}(\avsen{f}) = 2 K p (1-p)
\end{equation*}
\label{theo:ExpectedSensitivityB}
\end{theo}
For the random variable $F^\text{fixed}_{K,p}$ we will now proof the following theorem: 
\begin{theo}
Let the random variable $F^\text{fixed}_{K,p}$ be defined as above: 
\begin{equation*}
\EX_{F^\text{fixed}_{K,p}}(\avsen{f}) = \frac{2^{K+1} K p (1-p)}{(2^K-1)}.
\end{equation*}
\label{theo:ExpectedSensitivity}
\label{cor:01}
\end{theo}
\begin{proof}
To find  $\EX_{F^\text{fixed}_{K,p}}(\avsen{f})$ we will first consider the random variable $F_{K,t}:\Omega \rightarrow {\cal F_K}$
where $t \in \{0,1, \cdots, 2^K \}$
and the probability of a function is given by
\begin{equation*}
P( F_{K,t} = f ) = \begin{cases}
		\frac{1}{\binom{2^K}{t}} & \text{if}~ \wH(\vect{t}_f)=t \\
		0 & \text{else}
	      \end{cases}.
\end{equation*} 

Consider the Boolean functions as functions into $\field{R}$ by identifying $0,1 \in \GF$ with $0,1 \in \field{R}$.
Then we get or the function $f$:
\begin{align*}
\avsen{f} 	&= 2^{-K}\sum_{\vect{w} \in \GF^K} \setcard{ \set{ i \,|\, f(\vect{w}) \neq f(\unitvect{i} \oplus \vect{w}),\, i=1,\dots, K }} \\ 
 	  	&= 2^{-K}\sum _{\vect{w} \in \GF^K} \sum_{i=1}^K ( f(\vect{w}) - f(\vect{w} \oplus \unitvect{i}) )^2 \\
		&= 2^{-K}\sum _{\vect{w} \in \GF^K} \sum_{i=1}^K ( f(\vect{w}) + f(\vect{w} \oplus \unitvect{i}) - 2 f(\vect{w})f(\vect{w} \oplus \unitvect{i})).\\
\end{align*}
where $\unitvect{i}$ again denotes the unit vector with $i$th component set to $1$. 
Hence by the linearity of the expectation
\begin{equation}
\begin{split}
\EX_{F_{K,t}}(\avsen{f}) = 2^{-K} \sum_{\vect{w} \in  \GF^K} \sum_{i=1}^K &\left( \EX_{F_{K,t}}(f(\vect{w})) + \EX_{F_{K,t}}(f(\vect{w} \oplus \unitvect{i})) \right. \\ 
		  & - \left.2 \EX_{F_{K,t}}( f(\vect{w}) f(\vect{w} \oplus \unitvect{i}) ) \right).
\end{split}
\label{eq:proof00}
\end{equation}

Now we form a matrix with the truth tables of all
functions with Hamming weight $t$ as column vectors: 
\begin{equation*}
M = \begin{pmatrix}\indexedvect{c}{1}, & \indexedvect{c}{2}, \cdots ,& \indexedvect{c}{\binom{2^k}{t}} \end{pmatrix} \text{where $\indexedvect{c}{i} \in \GF^{2^n}$}
\end{equation*}
$M$ has exactly $\binom{2^K}{t}$ columns and $2^K$ rows. Each entry $M_{i,j}$ in the $i$th 
row and $j$th column equals the value of function $f_j$ given the
binary representation of $i$ as input.

Hence $\EX_{F_{K,t}}(f(\vect{w}))$ is determined by the number of $1$ in the row associated with $\vect{w}$ 
divided by the length of the row.
Consider an arbitrary row $i$. This row has a one at position $j$ if the corresponding column $\indexedvect{c}{j}$ 
has a one at position $i$. But there are $\binom{2^K-1}{t-1}$ column vectors with a $1$ at position $i$. 
It follows:
\begin{equation}
\forall \vect{w} \in \GF^K:\quad \EX_{F_{K,t}}(f(\vect{w})) = \frac{\binom{2^K -1}{t -1}}{\binom{2^K}{t}} 
				    = \frac{t}{2^{K}}.
\label{eq:proof01}
\end{equation}
As this holds for all $\vect{w}$, we have 
\begin{equation}
\forall \vect{w},\unitvect{i} \in \GF^K:\quad \EX_{F_{K,t}}( f(\vect{w} \oplus \unitvect{i})) = \frac{t}{2^{K}}.
\label{eq:proof02}
\end{equation}

To find an expression for $\EX_{F^\text{fixed}_{K,p}}( f(\vect{w}) f(\vect{w} \oplus \unitvect{i}) )$
we consider two arbitrary rows   $l,m$ ($l \neq m$). 
Define the following sum:
\begin{equation*}
\gamma_{l,m} = \sum_{i=1}^{\binom{K}{t}} M_{l,i} M_{m,i}.
\end{equation*}
Obviously $M_{l,i} M_{m,i} = 1$ only if we have a $1$ in both rows at position $i$.
This means for the column vectors $\indexedvect{c}{i}$ of $M$, 
we have $\indexedvectcomp{c}{i}{l} = \indexedvectcomp{c}{i}{m} = 1$. But there 
are exactly $\binom{2^K-2}{t-2}$ such column vectors in $M$. 
Therefore we have 
\begin{equation*}
\forall l,m, l \neq m: \gamma_{l,m} = \binom{2^K-2}{t-2}.
\end{equation*} 
As $\vect{w} \neq \vect{w} \oplus \unitvect{i}$ for all $\vect{w},\unitvect{i}$ it follows:
\begin{equation}
\EX_{F_{K,t}}( f(\vect{w}) f(\vect{w} \oplus \unitvect{i}) )= \frac{ \binom{2^K-2}{t-2} }{\binom{2^K}{t}} 
								  = \frac{t(t-1)}{2^K(2^K-1)}.
\label{eq:proof03}
\end{equation}
Hence substituting Equations \eqref{eq:proof01}, \eqref{eq:proof02} and \eqref{eq:proof03} into Equation \eqref{eq:proof00} 
leads to
\begin{equation*}
\EX_{F_{K,t}}(\avsen{f}) = 
\frac{K (2^K -t) t}{2^{K-1} (2^K -1)}.
\end{equation*}
Finally the claimed expression for $\EX_{F^\text{fixed}_{K,p}}(\avsen{f})$ can be obtained from the above equation
by a substitution of $t$: $t\rightarrow p 2^K$.
\end{proof}
\label{comment}
It should be noted, that the Theorems \ref{theo:01} and \ref{theo:ExpectedSensitivityB} can be proved using in a similar way.
Also worth noting is the fact, that if the functions are chosen according $F_K,F^\text{fixed}_{K,p}$ or $F_{K,p}$ 
the expectation of the sensitivity of a fixed vector $\vect{w}$ (namely the expectation of $s_f(\vect{w})$) is independent of $\vect{w}$
(see Equation \eqref{eq:proof00},\eqref{eq:proof01}, \eqref{eq:proof02} and \eqref{eq:proof03}). Hence the following lemma 
holds
\begin{lemma}
If $F=F_K,F^\text{fixed}_{K,p}$ or $F_{K,p}$, then
\begin{equation*}
\forall \vect{w},\vect{v} \in \GF^K: \EX_{F}(s_f(\vect{w})) =  \EX_{F}(s_f(\vect{v}))
\end{equation*}
\label{lemma:01}
\end{lemma}

Before proceeding to the next section, it should be noted, that using the same arguments as 
in the proof of Theorem \ref{theo:ExpectedSensitivity}, we can also prove the expectation of average
sensitivity of order $l$ , defined as
\begin{equation*}
s^{(l)}(f) = 2^{-K} \sum_{\vect{w} \in \GF^K} \setcard{ \set{ \vect{x} \in \GF^K | \wH(\vect{x}) = l \,\text{and}\, f(\vect{w}) \neq f(\vect{w} \oplus \vect{x}) }}.
\end{equation*}
In this case,  instead of summing up all unit vectors in Equation \eqref{eq:proof00}, we sum up all 
vectors of Hamming weight $l$. As the equations \eqref{eq:proof01} and \eqref{eq:proof03} hold for all $\vect{w} \in \GF^K$
we conclude that
\begin{equation*}
{\mathbb E}( s^{(l)}(F^\text{fixed}_{K,p}) ) = \binom{K}{l} \frac{2^{K+1} p (1-p)}{(2^K-1)}
\end{equation*}
and by similar arguments 
\begin{equation*}
{\mathbb E}( s^{(l)}(F_{K,p}) ) = \binom{K}{l} 2 p (1-p)
\end{equation*}
respectively 
\begin{equation*}
{\mathbb E}( s^{(l)}(F_{K}) ) =\frac{1}{2} \binom{K}{l}.
 \end{equation*}

\section{ Extending Lynch's analysis}
\label{sec:03}
As already mentioned James Lynch gave a very general analysis of 
randomly constructed Boolean networks (see \cite{Lynch2005}). Before stating his results we give a formal definition
for Boolean networks 
A Boolean network \BN is a 4-tuple $(V,E,\tilde{F},\vect{x})$ where 
$V = \{1,...,N\}$ is a set of natural numbers,  
$E$ is a set of labeled edges on $V$, 
$\tilde{F}=\{f_1, ..., f_N\}$ is a ordered set of Boolean functions such that
for each $v \in V$ the number of arguments of $f_v$ is the \emph{in-degree} of $v$ in $E$,
these edges are labeled with $1,... , \text{\emph{in-degree}}(v)$,
and $\vect{x} = ( \vectcomp{x}{1}, \dots \vectcomp{x}{n} ) \in \GF^N$. 
Suppose that a vertex $i$ has $K_i$ in-edges from vertices $ v_{i,1},\dots, v_{i,K_i}$.  
For $\vect{y} \in \GF^N$ we define 
\begin{equation*}
{\bf B}(\vect{y}) = \left( f_1( \vectcomp{y}{v_{1,1}}, \dots, \vectcomp{y}{v_{1,K_1}} ), \dots, f_N( \vectcomp{y}{v_{N,1}}, \dots, \vectcomp{y}{v_{N,K_N}} )\right).
\end{equation*} 
The state of \BN at time 0 is called the {\em initial state} $\vect{x}$, so we define $\vect{B}^0(\vect{x}) = \vect{x}$.
For time $t \geq 1$ the state is inductively defined as $\vect{B}^{t}(\vect{x}) = \vect{B}( \vect{B}^{t-1}(\vect{x}))$.
Hence we can in interpret $V$ as set of gates, $E$ and $\tilde{F}$ describes their functional dependence and 
$\vect{x}$ is the networks initial state. 

Assume some ordering $f_1, f_2, ...$ 
on the set of all Boolean functions $\cal F$, where each function $f_i$ depends on $K_i$ arguments.
Further
a random variable $F:\Omega \rightarrow {\cal F}$ with 
probabilities $p_i = P(F = f_i)$ 
such that $\sum_{i=i}^\infty p_i = 1$ and $\sum_{i=1}^\infty p_i K_i^2 < \infty$.
Now a random Boolean network consisting of $N$ {\em gates} is constructed as follows: For each gate
a Boolean function is chosen independently, where the probability of choosing $f_i$ is given by $p_i$. 
Suppose a function $f$ was chosen that has $K$ arguments, these arguments are chosen at random from all
$\binom{N}{K}$ equally likely possibilities. At last an initial state is chosen at random from 
the set on all equally likely states.
If the Boolean functions are chosen according to our previously defined random variable $F_K$ we
will call this networks $NK$-Networks with connectivity $K$. If the functions are chosen according to $F^\text{fixed}_{K,p}$ or $F_{K,p}$ we
will call this networks {\em biased random Boolean networks} with connectivity $K$ and fixed bias $p$ respectively mean bias $p$.

Let us now state Lynch's results. His analysis depends on a parameter $\field{R} \ni \lambda \geq 0$ depending only 
on the functions and their probabilities. We will 
define $\lambda$ later in Definition \ref{def:lambda}.
First we have to state Lynch's definition of {\em freezing} and {\em ineffective} gates:
\begin{dfn}[Lynch \cite{Lynch2005} Definition 1 Item 2 and 5]\hfill\\
Let $\vect{x} \in \GF^N$ and $v \in V$.
\begin{enumerate}
\item 
Gate $v$ freezes to $\vect{y} \in \GF^N$ in $t$ steps on input $\vect{x}$ if $\vect{B}_v^{t^\prime}(\vect{x}) = \vect{y}$ 
for all $t^\prime \geq t$. 
\item
Let $\unitvect{i} \in \GF^n$.\\
A gate $v$ is $t$-ineffective at input $\vect{x} \in \GF^K$ if $\vect{B}^t(\vect{x}) = \vect{B}^t(\vect{x} \oplus \unitvect{v})$.
\end{enumerate} 
\end{dfn}
Now we will state the main result.
\begin{theo}[Lynch \cite{Lynch2005} Theorem 4 and 6]\hfill\\
Let $\alpha$, $\beta$ be positive constants satisfying $2 \alpha \log \delta + 2 \beta < 1$ and  
$\alpha \log \delta < \beta$ where $\delta = {\mathbb E}(K_i)$.
\begin{enumerate}
\item  There is a constant $r$ such that for all $\vect{x} \in \GF^N$ 
\begin{equation*}
\lim_{n \rightarrow \infty} P(\text{$v$ is ineffective in $\alpha \log$ N steps}) = r
\end{equation*}
When $\lambda \leq 1$, $r=1$ and when $\lambda > 1$ , $r<1$. 
\item There is a constant $r$ such that for all $\vect{x} \in \GF^N$ 
\begin{equation*}
\lim_{n \rightarrow \infty} P(\text{$v$ is freezing in $\alpha \log$ N steps}) = r
\end{equation*}
When $\lambda \leq 1$, $r=1$ and when $\lambda > 1$ , $r<1$. 
\footnote{Please note that we here state a slightly weaker result than in the original analysis.}
\end{enumerate} 
\end{theo} 
The above theorem shows that if $\lambda \leq 1$ almost all gates are freezing and ineffective and otherwise not. 
The next corollary gives us more information what happens if $\lambda > 1$:
\begin{corollary}[Lynch \cite{Lynch2005} Corollary 3 and Corollary 6]
Let $\lambda > 1$. For almost all random Boolean networks
\begin{enumerate}
\item if gate $v$ is not $\alpha \log N$-ineffective, there is a positive constant $W$ such that for
$t \leq \alpha \log N$, the number of gates affected by $v$ at time $t$ is asymptotic to $W\lambda^t$,
\item if gate $v$ is not freezing in $\alpha \log N$ steps , there is a positive constant $W$ such that for
$t \leq \alpha \log N$, the number of gates that affect $v$ at time $t$ is asymptotic to $W\lambda^t$.
\end{enumerate}
\end{corollary} 

Now we will state the definition of $\lambda$ for Boolean networks:
\begin{dfn}[Lynch \cite{Lynch2005}, Definition 4]
\label{def:lambda}
Let $f$ be a Boolean function of $K$ arguments. For $i \in \set{1,\dots, K}$, we say that argument $i$ directly affects
$f$ on input $\vect{w} \in \GF^K$ if $f(\vect{w}) \neq f( \vect{w} \oplus \unitvect{i})$. Now put $\gamma(f,\vect{w})$ as the number of
$i$'s that directly affect $f$ on input $\vect{w}$. 
Given a constant $a \in [0,1]$, we define 
\begin{equation*}
\lambda = \sum_{i=1}^\infty p_i \sum_{\vect{w} \in \GF^{K_i} } \gamma(f_i,\vect{w})a^{\wH(\vect{w})}(1-a)^{K_i-\wH(\vect{w})}.
\end{equation*}
\end{dfn}
Obviously $\gamma(f,\vect{w})$ is identical to $s_f(\vect{w})$ which will be used instead
in the further discussion. 
The constant $a$ is the probability that a random gate is one (at infinite time) given that all gates
at time $0$ have probability $0.5$ of being one.
(see \cite[Definiton 2]{Lynch2005}).
Assume that we choose the functions according a random variable $F$ which should be either $F_K$, $F^\text{fixed}_{K,p}$ or $F_{K,p}$.
The functions are chosen out the set ${\cal F}_K$, we denote a function's probability with $p_f$.
It follows that 
\begin{align}
\lambda &= \sum_{\vect{w} \in \GF^{K} }a^{\wH(\vect{w})}(1-a)^{K-\wH(\vect{w})} \sum_f p_f s_f(\vect{w})\\
	&=  \sum_{\vect{w} \in \GF^{K} }a^{\wH(\vect{w})}(1-a)^{K-\wH(\vect{w})} \Expec{s_F(\vect{w})} \\
	&= \Expec{s_F(\vect{w})} \sum_{i=0}^K \binom{K}{i}a^{i}(1-a)^{K-i} \label{eq:lamda3}\\
	&= \Expec{s_F(\vect{w})} = \EX_{F}(\avsen{f})  
\end{align}
$\Expec{s_F(\vect{w})}$ denotes the expectation of the sensitivity for a fixed $\vect{w}$, 
Equation \eqref{eq:lamda3} follows from Lemma \ref{lemma:01}.
Therefore, together with Theorem \ref{theo:01} and Theorem \ref{cor:01} we proved
the following: 

\begin{theo}[Biased random Boolean networks]
For random Boolean networks, if
\begin{enumerate}
\item the functions are chosen according random variable $F_{K,p}$, it follows that
\begin{equation*}
\lambda =  2 K p (1-p),
\end{equation*}
\item the functions are chosen according random variable $F^\text{fixed}_{K,p}$, it follows that
\begin{equation*}
\lambda =  \frac{2^{K+1} K p (1-p)}{2^K-1}.
\end{equation*}
\end{enumerate}
\label{theo:02}
\end{theo}

As a special case of the above theorem we get (or by using Theorem \ref{theo:01}) 
\begin{theo}[$NK$-Networks]
In random Boolean networks, where the functions are chosen according to the random variable $F_{K}$
\begin{equation*}
\lambda = \frac{K}{2}.
\end{equation*}
\end{theo}

\section{Discussion}
The results about $NK$-Networks are consistent with experimental results. In fact if $K\leq2$ 
almost all networks  almost all gates are freezing and
almost all gates are ineffective and otherwise not (see \cite{Kauffman1974}). 

Obviously, the border between the ordered and disordered phase is given by $\lambda = 1$.
The resulting phase diagram for biased random Boolean networks, where the functions
are chosen according to $F^\text{fixed}_{K,p}$ and $F_{K,p}$ is shown in Figure \ref{fig:01}.
It it interesting to note that if the functions are chosen with fixed bias, 
then also Boolean networks with connectivity $K=2$ can become unstable.
This conclusion can be drawn from Lynch's original result already. As mentioned in the 
introduction, he showed for $K=2$, that $\lambda > 1$ if the probability of choosing a non-constant non-canalizing function, namely 
the XOR or the inverted XOR function, is larger than the probability of choosing a constant function.
For example if the bias is $0.5$, the probability of choosing
a constant function is zero, whereas both XOR and inverted XOR function have probability greater zero,
hence $\lambda > 1$.

\begin{figure}[ht]
\centerline{{\includegraphics[width=0.8\columnwidth]{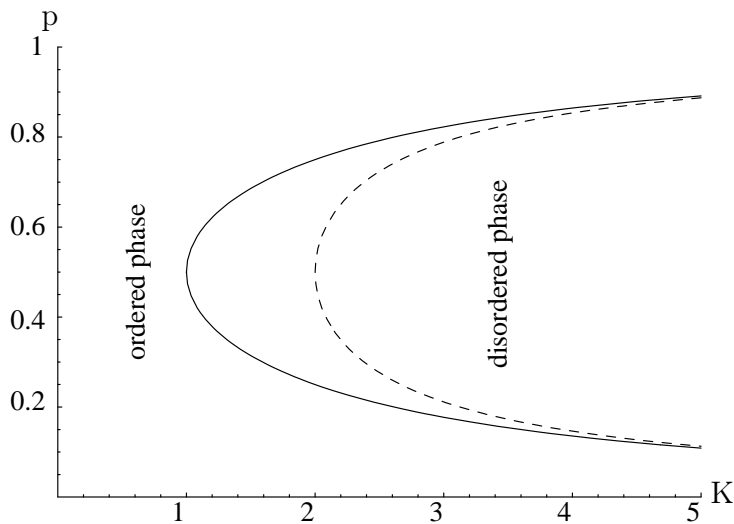}} }
\caption{Phase diagram for biased random networks: Functions chosen according $F_{K,p}$  (dashed) and $F^\text{fixed}_{K,p}$ (solid)}
\label{fig:01}
\end{figure}

It is interesting to compare our results with 
previous results obtained first by Derrida and Pomeau using the so called {\em annealed approximation}
(see \cite{Derrida1986}). In their \emph{annealed model} the functions and connections are chosen at random 
at each time step. Considering two instances of the same annealed network starting in two 
randomly chosen initial states $\vect{s}_1(0), \vect{s}_2(0)$ they show that 
\begin{equation*}
    \lim_{N \rightarrow \infty} \lim_{t \rightarrow \infty} \frac{\dH(\vect{s}_1(t), \vect{s}_2(t)) }{N} = c
\end{equation*}
where $c=1$ if 
\begin{equation*}
    2 K p (1-p) \leq 1 
\end{equation*}
and $c\leq 1$ otherwise. It is remarkable that the two models behave similar, but it 
is unclear whether this holds in general.

\section{Acknowledgement}
We would like to thank our colleges Georg Schmidt and Stephan Stiglmayr for proofreading and 
Uwe Schoening for useful hints.

\end{document}